\documentclass[10pt]{article}
\pdfoutput=1
\usepackage[utf8]{inputenc}
\usepackage[english]{babel}
\usepackage[T1]{fontenc}
\usepackage[intlimits]{amsmath}
\usepackage{amssymb,amsthm,amsfonts,graphicx,parskip,setspace,float,caption,algorithm,algpseudocode,enumerate,thmtools,thm-restate,geometry}
\usepackage[affil-it]{authblk}
\usepackage[pdftex,bookmarks,pagebackref, plainpages=false, pdfpagelabels=true]{hyperref}
\usepackage[inline]{enumitem}  
\usepackage[dvipsnames,svgnames,x11names]{xcolor}
\usepackage[capitalise]{cleveref}

\usepackage[title,toc]{appendix}

\newtheorem{theorem}{Theorem}
\newtheorem{lemma}{Lemma}
\theoremstyle{remark}

\numberwithin{equation}{section}

\DeclareMathOperator{\Q}{Q}

\newcommand{\ket}[1]{\mathinner{\left \lvert #1 \right\rangle}}

\newcommand\lr[1]{\left( #1 \right)}
\newcommand\lrv[1]{\left|  #1 \right|}

\newcommand\lrb[1]{\left\lbrace #1 \right\rbrace}

\newcommand{\mbb}[1]{\mathbb{#1}}

\newcommand{\mc}[1]{\mathcal{#1}}

\colorlet{blueish}{NavyBlue!75} 
\hypersetup{
    bookmarksnumbered=true, 
    unicode=false, 
    pdfstartview={FitH}, 
    pdftitle={Improved Algorithm and Lower Bound for Variable Time Quantum Search}, 
    pdfauthor={Andris Ambainis, Martins Kokainis, Jevgenijs Vihrovs}, 
    pdfsubject={}, 
    pdfcreator={}, 
    pdfproducer={}, 
    pdfkeywords={}, 
    pdfnewwindow=true, 
    colorlinks=true, 
    urlcolor=WildStrawberry, 
    linkcolor=blueish, 
    citecolor=ForestGreen, 
    filecolor=BrickRed 
}

\date{}
\title{Improved Algorithm and Lower Bound for Variable Time Quantum Search}
\author{Andris Ambainis}
\author{Martins Kokainis}
\author{Jevg\=enijs Vihrovs}
\affil{Center for Quantum Computer Science, Faculty of Computing,
University of Latvia} 
\begin{document}
\maketitle

\begin{abstract}
We study variable time search, a form of quantum search where queries to different items take different time. 
Our first result is a new quantum algorithm that performs variable time search with complexity $O(\sqrt{T}\log n)$ where $T=\sum_{i=1}^n t_i^2$ with $t_i$ denoting the time to check the $i^{\rm th}$ item. Our second result is a quantum lower bound of $\Omega(\sqrt{T\log T})$. Both the algorithm and the lower bound improve over previously known results by a factor of $\sqrt{\log T}$ but the algorithm is also substantially simpler than the previously known quantum algorithms.
\end{abstract}

\section{Introduction}

We study variable time search \cite{ambainis2010variable}, a form of quantum search in which the time needed for a query depends on
which object is being queried. Variable time search and its generalization, variable time amplitude \cite{ambainis2012variable}
amplification, are commonly used in quantum algorithms. For example,
\begin{itemize}
\item
Ambainis \cite{ambainis2012variable} used variable time amplitude amplification to improve the running time of HHL quantum algorithm for solving systems of linear equations  \cite{harrow2009quantum} from $\widetilde{O}(\kappa^2)$ (where $\kappa$ is the condition number of the system) to $\widetilde{O}(\kappa^{1+o(1)})$ in different contexts;
\item
Childs et al.~\cite{childs2017quantum} used variable time amplitude amplification to design a quantum algorithm for solving systems of linear equations with an exponentially improved dependence of the running time
on the required precision;
\item 
Le Gall \cite{legall2014improved} used variable time search to construct the best known quantum algorithm for triangle finding, with a running time $\widetilde{O}(n^{5/4})$ where $n$ is the number of vertices;
\item
De Boer et al.~\cite{boer2018attacks} used variable time search to optimize the complexity of quantum attacks against a post-quantum cryptosystem;
\item
Glos et al.~\cite{glos2021quantum} used variable time search to develop a quantum speedup for a classical dynamic programming algorithm.
\item Schrottenloher and Stevens \cite{Schrottenloher2022} used variable time amplitude amplification to transform a classical nested search into a quantum
algorithm, with applications to quantum attacks on AES.
\end{itemize}

In those applications, the oracle for the quantum search is a quantum algorithm whose running time depends on the item that is being queried. 
For example, we might have a graph algorithm that uses quantum search to find a vertex with
a certain property and the time $t_v$ to check the property may depend on the degree of the vertex $v$.  

In such situations, using standard quantum search would mean  using maximum time $t_{\max} = \max_v t_v$ for each query.
If most times $t_v$ are substantially smaller,
this results in suboptimal quantum algorithms. 

A more efficient strategy is to use the variable time quantum search algorithm \cite{ambainis2010variable}. 
It has two variants: the ``known times'' variant when times $t_v$ for checking various $v$ are known in advance and can be used to design the algorithm and the ``unknown times'' variant in which $t_v$ are only discovered when running the algorithm. In the ``known times'' case, VTS (variable time search) has complexity $O(\sqrt{T})$ where $T=\sum_{v} t_v^2$ and there is a matching lower bound \cite{ambainis2010variable}.

For the ``unknown times'' case, the complexity of the variable time search increases to $O(\sqrt{T} \log^{1.5} T)$ and the quantum algorithm becomes substantially more complicated. 
Since almost all of the applications of VTS require the ``unknown times'' setting, it may be interesting to develop a simpler quantum algorithm. 

In more detail, the ``unknown times'' search works by first running the query algorithm for a small time $T_1$ and then amplifying $v$ for which the query either returns a positive result or does not finish in time $T_1$. This is followed by running the query algorithm for longer time $T_2$, $T_3$, $\ldots$ and each time, amplifying $v$ for which the query either returns a positive result or does not finish in time $T_i$. To determine the necessary amount of amplification, quantum amplitude estimation is used. This results in a 
complex algorithm consisting of interleaved amplification and estimation steps. This complex structure contributes to the complexity of the algorithm, via log factors and may also lead to large constants hidden under the big-$O$.

In this paper, we develop a simple algorithm for variable time search that uses only amplitude amplification. Our algorithm achieves the complexity of $O(\sqrt{T} \log n)$ 
where $T$ is an upper bound for $\sum_v t_v^2$ provided to the algorithm. (Unlike in the ``known times'' model, we do not need to provide $t_1, \ldots, t_n$ but only an estimate for $T$.) This also improves over the previous algorithm by a $\sqrt{\log}$ factor.

To summarize, the key difference from the earlier 
algorithms 
\cite{ambainis2010variable,ambainis2012variable} is that the earlier algorithms would use amplitude estimation (once for each amplification step) to determine the optimal schedule for amplitude amplification for
this particular $t_1, \ldots, t_n$.  
In contrast, we use one fixed schedule for amplitude amplification (that depends only on the estimate for $T$ and not on $t_1, \ldots, t_n$). While this schedule may be slightly suboptimal, the losses from it being suboptimal are less than savings from not performing multiple rounds of amplitude estimations. This also leads to the quantum algorithm being substantially simpler.

Our second result is a lower bound of $\Omega(\sqrt{T \log T})$, showing that a complexity of $\Theta(\sqrt{T})$ is not achievable. The lower bound is by creating a query problem which can be solved by variable time search and using the quantum adversary method to show a lower bound for this problem.
In particular, this proves that ``unknown times'' search is more difficult than ``known times'' search (which has the complexity of $\Theta(\sqrt{T})$).

\section{Model, definitions, and previous results}

We consider the standard search problem in which the input consists of variables $x_1, \ldots, x_n\in\{0, 1\}$ and the task is to find $i:x_i=1$ if such $i$ exists.

Our model is a generalization of the usual quantum query model. We model a situation when the variable $x_i$ is computed by an algorithm $Q_i$ which is initialized in the state $\ket{0}$ and, after $t_i$ steps, outputs the final state $\ket{x_i}\ket{\psi_i}$ for some unknown $\ket{\psi_i}$. (For most of the paper, we restrict ourselves 
to the case when $Q_i$ always outputs the correct $x_i$.
The bounded error case is discussed briefly at the end of this section.)
In the first $t_{i}-1$ steps, $Q_i$ can be in arbitrary 
intermediate states.

The goal is to construct an algorithm $A$ that finds $i:x_i=1$ (if such $i$ exists).
The algorithm $A$ can run   
$Q_i$ for a chosen $t$,
with $Q_i$ outputting $x_i$ if $t_i\leq t$ or * (an indication that the computation is not complete) if $t_i>t$.
The complexity of $A$ is the amount of time that is spent running the  algorithms $Q_i$. Transformations that does not involve running $Q_i$ do not count towards the complexity.

More formally, we assume that, for any $T$, there is a circuit $C_T$ which, on an input $\sum_{i=1}^n \ket{i} \otimes \ket{0}$ 
outputs
\[ \sum_{i=1}^n \ket{i} \otimes \ket{y_i} \otimes \ket{\psi_i} \]
where $y_i=x_i$ if $t_i\leq T$ and $y_i =*$ if $t_i>T$. 
The state $\ket{\psi_i}$ contains intermediate results of the computation and can be arbitrary.
An algorithm $A$ for variable time search consists of two types of transformations:
\begin{itemize}
\item
circuits $C_T$ for various $T$;
 \item
transformations $U_i$ that are independent of $x_1, \ldots, x_n$.
\end{itemize}
If there is no intermediate measurements, an algorithm $A$ is of the form  
\[ U_k C_{T_k} U_{k-1} \ldots U_1 C_{T_1} U_0 \]
and its complexity is defined as $T_1+T_2+\ldots+T_k$. 
In the general case, an algorithm is a sequence
\[ U_0, C_{T_1}, U_1, \ldots, C_{T_k}, U_k \]
with intermediate measurements. Depending on the outcomes of those measurements,
the algorithm may stop and output the result or continue with the next transformations.
The complexity of the algorithm is defined as $p_1 T_1 + \ldots + p_k T_k$ 
where $p_i$ is the probability that $C_{T_i}$ is performed. (One could also allow $U_i$ and $T_i$ to vary depending on the results of previous measurements but this will not be necessary for our algorithm.)

If there exists $i:x_i=1$,
$A$ must output one of such $i$ with probability at least 2/3. If $x_i=0$, $A$ must output ``no such $i$'' with probability at least 2/3.

{\bf Known vs.~unknown times.}
This model can be studied in two variants. In the ``known times'' variant, 
the times $t_i$ for each $i\in [n]$ are known in advance and can 
be used to design the search algorithm. In the ``unknown times'' variant,
the search algorithm should be independent of 
the times $t_i$, $i\in [n]$.

The complexity of the variable time search is characterized by the parameter
$T = \sum_{i=1}^n t_i^2$. We summarize the previously known results below.

\begin{theorem}
\label{thm:known}
\cite{ambainis2010variable,ambainis2012variable}
\begin{enumerate}[label=(\alph*)]
    \item
    {\bf Algorithm -- known times:} For any $t_1, \ldots, t_n$, there is a variable time search algorithm $A_{t_1, \ldots, t_n}$ with the complexity $O(\sqrt{T})$.
    \item
    {\bf Algorithm -- unknown times:} There is a variable time search algorithm $A$ with the complexity $O(\sqrt{T} \log^{1.5} T)$ for the case when $t_1, \ldots, t_n$ are not known in  advance.
    \item
    {\bf Lower bound -- known times.} For any $t_1, \ldots, t_n$ and any variable time search algorithm $A_{t_1, \ldots, t_n}$, its complexity must be $\Omega(\sqrt{T})$.
\end{enumerate}

\end{theorem}

Parts (a) and (c) of the theorem are from \cite{ambainis2010variable}. Part (b) is from \cite{ambainis2012variable}, specialized to the case of search.

In the recent years there have been attempts to reproduce and improve the aforementioned results by other means. In \cite{cornelissen_et_al:LIPIcs:2020:12694}, the authors obtain a variant of Theorem \ref{thm:known}(a) by converting the original algorithms into span programs, which then are composed and subsequently converted back to a quantum algorithm. More recently, \cite{Jeffery2022}  gives variable time quantum walk algorithm (which generalizes variable time quantum search) by employing a  recent technique of multidimensional quantum walks. While the focus of these two papers is on developing very general frameworks, our focus is on making the variable time search algorithm simpler.

Concurrently and independently of our work, a similar algorithm for variable time amplitude amplification was presented in \cite{Schrottenloher2022}, which also relies on recursive nesting of quantum amplitude amplifications. 

{\bf Variable time search with bounded error inputs.} We present our results for the case when the queries $Q_i$ are perfect (have no error) but our algorithm can be extended to the case if $Q_i$ are bounded error algorithms, at the cost of an extra logarithmic factor.

Let $k$ be the maximum number of calls to $C_T$'s in an algorithm $A$. Then, it suffices that each $C_T$ outputs a correct answer with a probability $1-o(1/k^2)$. 
This can be achieved by repeating $C_T$ $O(\log k)$ times and taking the majority of answers.

Possibly, this logarithmic factor can be removed using methods similar to ones for search with bounded error inputs in the standard (not variable time) setting \cite{Hoyer2003quantum}.

\section{Algorithm}

We proceed in two steps. We first present a simple algorithm for the case 
when a sufficiently good bound on the number of solutions $m=|i:x_i=1|$ are known (Section \ref{sec:known}). We then present an algorithm for the general case that calls the simple algorithm multiple times,
with different estimates for the parameter $\ell$ corresponding to $m$ (Section \ref{sec:unknown}).

Both algorithms require an estimate $T$ for which $\sum_{i=1}^n t_i^2 \leq T$, with the complexity depending on $T$.

\subsection{Tools and methods}

Before presenting our results, we describe the necessary background about quantum amplitude amplification \cite{BHMT02}.

{\bf Amplitude amplification -- basic construction.} 
Assume that we have an algorithm $A$ that succeeds with a small probability 
and it can be verified whether $A$ has succeeded. Amplitude amplification is a procedure for increasing the success
probability. Let 
\[ A\ket{0} = \sin \alpha \ket{\psi_{\text{succ}}} + \cos \alpha \ket{\psi_{\text{fail}}} .\]
Then, there is an algorithm $A(k)$ that involves $k+1$ applications of $A$ and $k$ applications of $A^{-1}$
such that  
\[ 
A(k)\ket{0} = \sin \lr{(2k+1) \alpha} \ket{\psi_{\text{succ}}} 
+ 
\cos \lr{(2k+1) \alpha} \ket{\psi_{\text{fail}}}.\] 
Knowledge of $\alpha$ is not necessary (the way how $A(k)$ is obtained from $A$ is independent of $\alpha$).

{\bf Amplitude amplification -- amplifying to success probability $1-\delta$.}
If $\alpha$ is known then one can choose $k = \lfloor \frac{\pi}{4 \alpha}\rfloor$ to amplify to a success
probability close to 1 (since $(2k+1) \alpha$ will be close to $\frac{\pi}{2}$).
If the success probability of $A$ is $\epsilon$, then $\sin \alpha \approx \sqrt{\epsilon}$ 
and $k \approx \frac{\pi}{4\sqrt{\epsilon}}$. 

For unknown $\alpha$, amplification to
success probability $1-\delta$ for any $\delta>0$ can be still achieved, via a more complex algorithm.  
Namely, for any $\epsilon, \delta\in(0, 1)$
and any $A$, one can construct an algorithm $A(\epsilon, \delta)$ such that:
\begin{itemize}
    \item $A(\epsilon, \delta)$ invokes $A$ and $A^{-1}$ $O(\frac{1}{\sqrt{\epsilon}}\log \frac{1}{\delta})$ times;
    \item If $A$ succeeds with probability at least $\epsilon$, $A(\epsilon, \delta)$ succeeds with probability at least $1-\delta$.
\end{itemize}
To achieve this, we first note that performing $A(k)$
for a randomly chosen $k\in\{1, \ldots, M\}$ for
an appropriate $M=O\left(\frac{1}{\sqrt{\epsilon}}\right)$ and measuring the final state gives a success probability that is close to 1/2 (as observed in the proof of Theorem 3 in \cite{BHMT02}). 
Repeating this procedure 
$O\left(\log \frac{1}{\delta}\right)$ times
achieves the success probability of at least $1-\delta$.

\subsection{Algorithm with a fixed number of stages}
\label{sec:known}

Now we present an informal overview of the algorithm when tight bounds on the number of solutions $m=|i:x_i=1|$ is known.
We will define a sequence of times $T_1, T_2, \ldots$ and procedures $A_1, A_2, \ldots$. We choose $T_1= 3 \sqrt{T/n}$ (this ensures that at most $n/9$ of indices $i\in[n]$ have $t_i\geq T_1$) and $T_2=3 T_1$, $T_3 = 3 T_2$, $\ldots$ until $d$ for which $T_d \geq \sqrt{T}$. The procedure $A_1$ creates the superposition  $\sum_{i=1}^n \frac{1}{\sqrt{n}} \ket{i}$ and runs the checking procedure $C_{T_1}$, obtaining state of the form
$\sum_{i=1}^n \frac{1}{\sqrt{n}} \ket{i, a_i}$,
where $a_i\in\{0, 1, *\}$, with $*$ denoting a computation that did not terminate. The subsequent procedures $A_j$ are defined as $A_j = C_{T_j} A_{j-1}(1)$, i.e., we first amplify the parts of the state with outcomes  1 or $*$ and then run the checking procedure $C_{T_j}$.

We express the final state of $A_{j-1}$ as  
\[ 
\sin \alpha_{j-1} \ket{\psi_{\text{succ}}} + \cos \alpha_{j-1} \ket{\psi_{\text{fail}}} ,
\]
where $\ket{\psi_{\text{succ}}} $ consists of those indices $i\in [n]$ which are either 1 or   are still unresolved $*$ (and thus have the potential to turn out to be `1').
Then  the amplitude amplification part  triples the angle $\alpha_{j-1}$, i.e., amplifies both the `good' and `unresolved' states by a factor of $\sin(3\alpha_{j-1})/\sin(\alpha_{j-1}) \approx 3$. We will show that  $\ell = \lceil \log_9 \frac{n}{m} \rceil$ stages are sufficient, i.e., the procedure $A_{\ell}$  the amplitude at the `good' states (if they exist) is sufficiently large.

We note that the idea of recursive tripling via amplitude amplification has been used in other contexts. It has been used to build an algorithm for bounded-error search in \cite{Hoyer2003quantum}; more recently, the recursive tripling trick has also been used in, e.g., \cite{chakraborty_et_al:LIPIcs.STACS.2022.20}. Furthermore, the repeated tripling of  the angle $\alpha$ also explains the scaling factor 3 when defining the sequence $T_1, T_2, T_3 \ldots$

A formal description follows.

We assume  an estimate $T \geq \sum_i t_i^2$ to be known and set
\begin{equation*}
	T_1= 3 \sqrt{T/n}, \  T_2 = 3T_1, \ \ldots, T_d = 3 T_{d-1},
\end{equation*}
with $ d \in \mbb N $ s.t.~$ T_{d-1 } < \sqrt T \leq T_d $  (equivalently, $ 9^{d-1} <  n \leq 9^d $). 

Let $ \mc M = \lrb{i\in[n] : x_i=1 }$, $m= \lrv{\mc M}$. 
We assume that we know $\ell$ for which $m$ belongs to the interval $\left [\frac{n}{9^\ell}, \frac{n}{9^{\ell-1}} \right)$ 
(so that $\ell = \lceil \log_9 \frac{n}{m} \rceil$).

Under those assumptions, we now describe a variable time search algorithm with parameters $T, \ell$.

\begin{algorithm}[H]
	\caption{VTS algorithm with a fixed number of stages
	}\label{alg:alg1}
	\begin{algorithmic}[1]
    \Statex {\bf Parameters: }{$ T $, $ n $, $\ell$, $\delta$.}
      \State{Run the amplified algorithm $A(0.04, \delta)$ where $A$ is the procedure defined below and we amplify the part of the state for the second register contains `1' 
      }
      \Procedure{$A$}{}
      \State{Run $A_\ell$}\Comment{(defined below)}
      \State{Run $C_{T_{\ell+1}}$ (or $C_{T_d}$ if $\ell=d$)}
      \EndProcedure
      \State Measure the state  
      \If {The second register is `1'}
        \State Output $i$ from the first register
        \Else 
        \State Output \texttt{No solutions.}
      \EndIf 
      \Procedure{$A_j$}{}\Comment{$j \in [d]$}
      \If{$j=1$}
        \State  Create the state $\sum_{i=1}^n \frac{1}{\sqrt{n}} \ket{i}$ 
        \State Run $C_{T_1}$, obtaining state of the form
        $\sum_{i=1}^n \frac{1}{\sqrt{n}} \ket{i, a_i}  $ 
where $a_i\in\{0, 1, *\}$.
      \Else
      \State Perform the amplified algorithm $A_{j-1}(1)$,
      amplifying the basis states with 1 or * in the second register 
      \If{$j<\ell$}
      \State  Run $C_{T_j}$.
      \EndIf 
      \EndIf
      \EndProcedure
  \end{algorithmic}
\end{algorithm}

\begin{lemma}\label{thm:proofOfAlg1} 
	\cref{alg:alg1}  with parameter $\ell = \lceil \log_9 \frac{n}{m} \rceil$ finds an index   $ i  \in \mc M$ with probability at least $1-\delta$ in time   $ O \lr{   \sqrt {\frac{T}{m}} 	\log\frac{n}{m} \log \frac{1}{\delta} } $.
\end{lemma}

\begin{proof}
	
	By $\mc S_j$ we denote the sets of those indices whose amplitudes  will be amplified   after running $A_j$, namely, the set of indices for which the query either returns a positive result or does not finish in time $T_j$:
	\begin{equation*}
		\mc S_j = 	\lrb{ i \in [n] \ : \   \lr{T_j < t_i} \lor \lr{t_i \leq T_j \land 
				x_i=1
			}  
		} ,
		\quad j=0,1,2,\ldots, d,
	\end{equation*}
	where $T_0 := 0$. We note that the sets $\mc S_j$ form a decreasing sequence\footnote{Since   each  $i$ s.t.~$t_i \leq T_j \land  x_i=1$ either satisfies $t_i \leq T_{j-1} \land x_i=1$ or $t_i > T_{j-1}$; in both cases $i \in \mc S_{j-1}$.},
	i.e., 
	\begin{equation*}
		[n] = 
		\mc S_0 \supseteq \mc S_1 \supseteq \mc S_2 \supseteq \ldots \supseteq 
		\mc S_{d-1} \supseteq \mc S_d = \mc M.
	\end{equation*}
	We shall denote the cardinality of $\mc S_j $ by $s_j$; then 
	\begin{equation*}
		n = s_0  \geq s_1 \geq \ldots \geq s_d = m.
	\end{equation*}
	
	We express the final state of $A_j$ as 
\[ \sin \alpha_j \ket{\psi_{\text{succ}, j}} + \cos \alpha_j \ket{\psi_{\text{fail}, j}} \]
where $\ket{\psi_{\text{succ}, j}}$ consists of basis states with $\ket{i}$, $i\in \mc  S_j$,
in the first register and $\ket{\psi_{\text{fail}, j}}$ consists of basis states with $\ket{i}$, $i\notin \mc S_j$, in the first register.

	We begin by describing how the cardinality of $\mc S_j$ is related to the amplitude $\sin \alpha_j$ (the proof is deferred  to \cref{sec:appA}).
	
	\begin{restatable}{lemma}{recursiveamplitude}\label{thm:recursiveamplitude}
		For all $j=1,2,\ldots, \ell$, 
		\begin{equation}\label{eq:nextSinJ}
			\sin^2 \alpha_{j} 
			=
			\frac{s_{j}}{n}
			\prod_{k=1}^{j-1}  \lr{\frac{\sin (3\alpha_{k})}{\sin\alpha_{k}}}^2.
		\end{equation}
		Moreover, for any $i \in \mc S_j$, the amplitude at $\ket{i,1}$  (or $\ket{i,*}$, if $t_i > T_j$)  equals $ \frac{\sin \alpha_j}{\sqrt{s_j}} $.
	\end{restatable}
	
	\cref{eq:nextSinJ} and the trigonometric identity
	\begin{equation*}
		\sin(3\alpha) 
		=
		(3-4\sin^2\alpha)\sin \alpha
	\end{equation*}
	allows 
	to obtain (for $j=1,2,\ldots,\ell$)
	\begin{equation}\label{eq:sinRatioLB}
		\frac{\sin (3\alpha_{j}) }{\sin\alpha_{j} }
		=
		3-4\sin^2\alpha_{j}
		=
		3-
		\frac{4s_j \cdot 9^{j-1} }{n}
		\prod_{k=1}^{j-1}  
		\lr{\frac{\sin (3\alpha_{k})}{3\sin\alpha_{k}}}^2
		\geq 
		3- \frac{4s_j  }{n} \cdot 9^{j-1},
	\end{equation}
	where the inequality is justified by the observation $\lrv{\frac{\sin (3\alpha)}{3\sin \alpha}} \leq 1 $. This allows to estimate 
	\begin{equation}\label{eq:sinAlphaL}
		\sin \alpha_\ell
		=
		\sqrt{\frac{s_\ell}{n}}
		\prod_{j=1}^{\ell-1}
		\frac{\sin (3\alpha_{j})}{\sin\alpha_{j}}
		\geq 
		3^{\ell-1} \sqrt{\frac{s_\ell}{n}}
		\prod_{j=1}^{\ell-1}
		\lr{ 1 - \frac{4 s_j}{27 n} \cdot 9^j },
	\end{equation}
	as long as each factor on the RHS is positive.
		We argue that it is indeed the case; moreover, the whole product is lower-bounded by a constant  (the proof is deferred  to \cref{sec:appA}):
	\begin{restatable}{lemma}{cOneThree}\label{thm:C1-3}
		\begin{enumerate}[label=\textbf{C-\arabic*}]
			\item\label{enum:1} Each factor on the RHS  of \eqref{eq:sinAlphaL} is   positive: $ \frac{9^j s_j}{n} \leq  \frac{9}{4}   $, thus 
			\begin{equation*}
			  \lr{ 1 - \frac{4 s_j}{27 n} \cdot 9^j } \geq \frac{2}{3},
			  \quad\text{for all }
			  j \in [\ell-1].
			\end{equation*} 
			\item\label{enum:2} The product $\prod_{j=1}^{\ell-1}
			\lr{ 1 - \frac{4 s_j}{27 n} \cdot 9^j }$ is lower bounded by $2/3$.
			\item\label{enum:3} $9^\ell s_\ell \geq 9^\ell s_d \geq  n$.
		\end{enumerate}
	\end{restatable}

	From \eqref{eq:sinAlphaL} and \cref{thm:C1-3}  it is evident that
	$
	\sin \alpha_\ell \geq  \frac{2}{9} \sqrt{ \frac{9^\ell s_\ell}{n} }  \geq \frac{2}{9}
	$.

	However, after running $A_\ell$, there still could be some unresolved indices $i$ with $t_i > T_\ell$ and some of these unresolved indices may correspond to $x_i=0$. Our next argument is that running $C_{T_{\ell+1}}$, i.e., the checking procedure for $3 T_\ell$ steps,  resolves sufficiently many indices in $\mc M$. 
	This argument, however, is necessary only for $\ell < d$; for $\ell=d$, one runs $C_{T_d}$ instead of $C_{T_{\ell+1}}$ and the same estimate \eqref{eq:A1succProbLB} of the success probability applies, with $8m/9$ replaced by $m$. Also notice that in \cref{alg:alg1} we skipped running $C_{T_\ell}$ at the end of $A_\ell$ and immediately proceeded  with running $C_{T_{\ell+1}}$ instead. In the analysis, this detail is omitted for convenience (since it is equivalent to running $C_{T_j}$ at the end of each procedure $A_j$ and additionally running $C_{T_{\ell+1}}$ after $A_\ell$).

	By the choice of $\ell$ we have $\sqrt{\frac{T}{m}} \leq   T_\ell = \sqrt{ \frac{9^\ell T }{n} }$ and $T_{\ell+1} \geq 3 \sqrt{\frac{T}{m}}$. Notice that at most  $m/9$ of the indices $i \in [n]$ can satisfy $t_i^2 >  T_{\ell+1}^2 $ (otherwise, the sum over those indices already exceeds $ \frac{m}{9} \cdot  \frac{9T}{m} = T   $). Consequently, after running the checking procedure $C_{T_{\ell+1}}$, at  least $8m/9$ of the indices in $\mc M$  will be resolved to `1'.
		By \cref{thm:recursiveamplitude}, the amplitude at each of the respective states $\ket{i,1}$ is  equal to $\frac{\sin \alpha_\ell}{\sqrt{s_\ell}}$, therefore the probability to measure `1' in the second register is at least
	\begin{equation}\label{eq:A1succProbLB}
	    \frac{8m}{9} \cdot \frac{\sin^2 \alpha_\ell}{s_\ell}
	    \geq
	    \frac{8m}{9} \cdot
		 {\frac{9^{\ell} }{n}}
		\lr{\frac{1}{3}\prod_{j=1}^{\ell-1}
		\lr{ 1 - \frac{4 s_j}{27 n} \cdot 9^j }}^2
		\geq 
		\frac{8}{9} \cdot \lr{ \frac{ 2}{9}}^2>0.04
		,
	\end{equation}
	where the first inequality follows from \eqref{eq:sinAlphaL} and the second inequality is due to \ref{enum:2} and \ref{enum:3}.

We conclude that   the procedure
$ A$  
finds an index $i \in \mc M$ with probability at least $0.04$; 
its running time  is easily seen to be
\begin{equation*}
	T_{\ell+1} + T_{\ell} + 3 \lr{ T_{\ell-1} +  3 \lr{T_{\ell-2} + \ldots + 3 \lr{T_2 +  3 T_1 }  }  }
	=
	(3+\ell) T_\ell ,
\end{equation*}
which for our choice of $\ell$ is of order
\begin{equation*}
     O \lr{ 
		\log\frac{n}{m} \sqrt {9^\ell \frac{T}{n}}    } 
	=  O \lr{  
		\log\frac{n}{m} \sqrt {\frac{T}{m}}  } .
\end{equation*}
Use $O(\log \frac{1}{\delta})$ rounds  amplitude amplification to amplify the success probability of $ A$  to $1-\delta$, concluding the proof.
\end{proof}

\subsection{Algorithm for the general case}
\label{sec:unknown}

When the cardinality of $\lrv{\mc M}$ is not known in advance, we run \cref{alg:alg1} with increasing values of $\ell$ (which corresponds to exponentially decreasing guesses of  $m$) until either $i : x_i=1$ is found or we conclude that no such $i$ exists.  \cref{alg:alg1} also suffers the `souffl\'{e} problem' \cite{Brassard1997} in which iterating too much (choosing $\ell$ in \cref{alg:alg1} larger than its optimal value) may ``overcook'' the state and decrease the success probability. For this reason,  before running \cref{alg:alg1} with the next value of $\ell$, we re-run it with all the previous values of $\ell$ to ensure that the probability of running \cref{alg:alg1} with too large $\ell$ is small. This ensures that
the algorithm stops in time $ O \lr{   \sqrt {\frac{T}{m}} 	\log\frac{n}{m} } $ with high probability. 
Formally, we make the following claim:
 
\begin{restatable}{lemma}{algIIanalysis}\label{thm:alg2analysis}
    If $\mc M$ is nonempty,  
\cref{alg:alg2}  
finds an index   $ i  \in \mc M$ with probability at least $5/6$  with 
complexity $ O \lr{   \sqrt {\frac{T}{m}} 	\log\frac{n}{m} } $.
If $\mc M$ is empty,  \cref{alg:alg2}  outputs \texttt{No solutions.} with complexity $ O \lr{   \sqrt {T} 	\log n } $.
\end{restatable}

\begin{algorithm}[H]
	\caption{VTS algorithm for arbitrary number of solutions $m$  
	}\label{alg:alg2}
  \begin{algorithmic}[1]
    \Statex {\bf Parameters: }{$ T $, $ n $.}
    \Statex Let $B_k $ stand for \cref{alg:alg1} with parameters $T$, $n$, $k$ and $\delta = 1/6$.
      \For{$j= 1,2,\ldots,d$}
          \For{$k = 1,2,\ldots,j$} 
          \State Run $B_k$
            \State If $B_k$  
            returned $i\in \mc M$, output this $i$ and quit
          \EndFor
      \EndFor
      \State {Output \texttt{No solutions.}}
  \end{algorithmic}
\end{algorithm}

\begin{proof}[Proof of \cref{thm:alg2analysis}]
Let $\delta=1/6$; let us remark that each procedure $B_k$ runs in time $O\lr{k 3^k \sqrt{T/n}}$.

Let us consider the case when $ m=\lrv{\mc M}>0 $;  denote  $\ell := \lceil \log_9 \frac{n}{m} \rceil$. The probability of $ B_k$, $k \neq \ell$, finding an index $i \in \mc M$ is lower-bounded by 0; the probability of $ B_{\ell} $ finding an index $i \in \mc M$ is lower-bounded by $1-\delta$.

Hence, the total complexity of the algorithm stages $j=1,2,\ldots, \ell$, is of order 
\begin{equation*}
    \sqrt{\frac{T}{n}}\sum_{j=1}^\ell \sum_{k=1}^j k 3^k 
    =
    \sqrt{\frac{T}{n}}\sum_{j=1}^\ell (\ell+1-j) j 3^j
    \asymp
    \ell 3^\ell \sqrt{\frac{T}{n}}  .
\end{equation*}
and the last step $B_\ell$ finds $i \in \mc M$ with probability at least $1-\delta$.

With probability at most $\delta$, 
 the last step fails to find $i\in \mc M$, and then \cref{alg:alg2}
 proceeds with $j=\ell+1$ and
 runs the sequence $B_1$, $B_2$, \ldots, $B_\ell$, where the last step finds $i \in \mc M$ with (conditional) probability at least $1-\delta$ (conditioned on the failure to find $i\in \mc M$ in the previous batch).
The complexity of this part is of order
\begin{equation*}
    \delta  \sqrt{\frac{T}{n}} \lr{
    \sum_{j=1}^\ell j 3^j 
    }
    \asymp
     \delta  \sqrt{\frac{T}{n}} 
    \ell 3^\ell ,
\end{equation*}
where the $\delta$ factor reflects the fact the respective procedures are invoked with probability $\delta$.

With (total) probability at most $\delta^2$, 
the algorithm still has not found   $i\in \mc M$. Then \cref{alg:alg2}
 runs the sequence $B_{\ell+1}$, $B_1$, $B_2$, \ldots, $B_\ell$ (i.e., finishes with $j=\ell+1$ and continues with $j=\ell+2$), where the last step finds $i \in \mc M$ with (conditional) probability at least $1-\delta$.
The complexity of this part is of order
\begin{equation*} 
     \delta^2  \sqrt{\frac{T}{n}} \lr{
     (\ell+1)3^{\ell+1}
     +
    \ell 3^\ell
    }
    \asymp 
    \delta^2  \sqrt{\frac{T}{n}} (\ell+1)3^{\ell+1}
    .
\end{equation*}

With (total) probability at most $\delta^3$, 
the algorithm still has not found   $i\in \mc M$. Then \cref{alg:alg2} runs the sequence $B_{\ell+1}$,$B_{\ell+2}$, $B_1$, $B_2$, \ldots, $B_\ell$, where the last step finds $i \in \mc M$ with (conditional) probability at least $1-\delta$.
The complexity of this part is of order
\begin{equation*}
     \delta^3  \sqrt{\frac{T}{n}} \lr{
     (\ell+1)3^{\ell+1}
     +(\ell+2)3^{\ell+2}
     +
    \ell 3^\ell
    }
    \asymp 
    \delta^3  \sqrt{\frac{T}{n}} (\ell+2)3^{\ell+2},
\end{equation*}
and so on.

For $j=d$, the final batch $B_1, B_2, \ldots, B_\ell$  is invoked with probability at most $\delta^{d-\ell}$; with conditional probability at most $\delta$ we still fail to find $i\in \mc M$ and run the remaining sequence $B_{\ell+1}, \ldots, B_d$ (which can completely fail finding any $i\in \mc M$  as it has no non-trivial lower bounds on the success probability). The complexity of the latter sequence is of order
\begin{equation*}
     \delta^{d+1-\ell}  \sqrt{\frac{T}{n}} \lr{
     (\ell+1)3^{\ell+1}
     +(\ell+2)3^{\ell+2}
     +
     \ldots 
     +d 3^d
    }
    \asymp 
    \delta^{d+1-\ell}  \sqrt{\frac{T}{n}} \, d \, 3^{d}.
\end{equation*}

We see that \cref{alg:alg2} fails with probability at most $\delta^{d+1-\ell}$; since $\ell \leq d$, this is upper-bounded by $\delta=1/6$. The total complexity of the algorithm is of order
\begin{align*}
    & 3^\ell \sqrt{\frac{T}{n}}\lr{
    \ell 
    + 3\delta ^2 (\ell +1)
    + 9 \delta^3   (\ell +2)
    + \ldots 
    + (3 \delta)^{d-\ell} \cdot d \delta  
    }
    \\
    &
    <
 3^\ell \sqrt{\frac{T}{n}} \lr{
 \ell 
 + \ell \, 3\delta^2 \sum_{i=0}^\infty(3\delta)^i
 +  3\delta^2 \sum_{i=1}^\infty i (3\delta)^{i-1}
 }
    \\
    &
    =
 3^\ell \sqrt{\frac{T}{n}} \lr{
 \ell 
 + \ell  \frac{3\delta^2 }{1-3\delta}
 +  \frac{3\delta^2 }{(1-3\delta)^2}
 }
 \asymp 
  \ell 3^\ell \sqrt{\frac{T}{n}},
\end{align*}
since $3\delta=1/2$. Since $3^\ell \asymp \sqrt{\frac{n}{m}} $ and $\ell \asymp 	\log\frac{n}{m} $, we conclude that the complexity of the algorithm is  $ O \lr{   \sqrt {\frac{T}{m}} 	\log\frac{n}{m} } $, as claimed.

Let us consider the case when $\mc M $ is empty;  then  
with certainty
each $B_j$ fails to output any $i$, and \cref{alg:alg2} correctly outputs \texttt{No solutions}.  In this case, the complexity of the algorithm is of order
\begin{equation*}
    \sqrt{\frac{T}{n}}\sum_{j=1}^d \sum_{k=1}^j k 3^k 
    =
    \sqrt{\frac{T}{n}}\sum_{j=1}^d (d+1-j) j 3^j
    \asymp
    d 3^d \sqrt{\frac{T}{n}} \asymp \sqrt{T} \log n,
\end{equation*}
since $3^d \asymp \sqrt{n}$.
\end{proof}

\section{Lower bound}

For the improved lower bound, we consider a query problem which can be solved with variable time search.
Let $g : \{0,1,\star\}^m \to \{0,1\}$ be a partial function defined on the strings with exactly one non-$\star$ value, which is the value of the function.
The function $f$ we examine then is the composition of $\text{OR}_n$ with $g$. We note that $g$ is also known in the literature as $\text{pSEARCH}$, which has been used for quantum lower bounds in cryptographic applications \cite{BBHKLS18}.

For any $i \in [n]$, if the index of the non-$\star$ element in the corresponding instance of $g$ is $j_i \in [m]$, then we can
find this value in $O(\sqrt{j_i})$ queries using Grover's search.
This creates an instance of the variable search problem with unknown times $t_i = \sqrt{j_i}$.
By examining only inputs with fixed $T = \sum_{i=1}^n t_i^2 = \sum_{i=1}^n j_i$ and the restriction of $f$ on these inputs $f_T$, we are able to prove a $\Omega(\sqrt{T \log T})$ query lower bound using the weighted quantum adversary bound \cite{Aar06}.
Since any quantum algorithm for the variable time search also solves $f_T$, this gives the required lower bound.

\begin{restatable}{theorem}{thmLB}\label{thm:lb}
Any algorithm that solves variable time search with unknown times $t_i$ requires time $\Omega(\sqrt{T \log T})$, where $T = \sum_{i\in[n]} t_i^2$.
\end{restatable}

We note that the lower bound of Theorem \ref{thm:lb} contains a factor of $\sqrt{\log T}$ while the upper bound of Lemma \ref{thm:alg2analysis} contains a factor of $\log n$.
There is no contradiction between these two results as the lower bound uses inputs with
$T=\Theta(n \log n)$ and for those inputs $\log T= (1+o(1)) \log n$.

\begin{proof}[Proof of Theorem \ref{thm:lb}]

Consider a partial function $f : D \to \{0,1\}$, where $D \subset \{\star,0,1\}^{n\times m}$, defined as follows.
An input $x \in D$ if for each $i \in [n]$ there is a unique $j \in [m]$ such that $x_{i,j} \neq \star$; denote this $j$ by $j_{x,i}$.
Then $f(x) = 1$ iff there exists an $i$ such that $x_{i,{j_{x,i}}} = 1$.

Suppose that $x$ is given by query access to $x_{i,j}$.
For any $i$, we can check whether $x_{i,j_{x,i}} = 1$ in $O(\sqrt{j_{x,i}})$ queries with certainty in the following way.
There is a version of Grover's search that detects a marked element out of $N$ elements in $O(\sqrt N)$ queries with certainty, if the number of marked elements is either $0$ or $1$ \cite{BHMT02}.
By running this algorithm for the first $N$ elements, where we iterate over $N = 1, 2, \ldots, 2^{\lceil \log_2{j_{x,i}}\rceil}$, we will detect whether $x_{i,j_{x,i}} = 1$ in $O(\sqrt{j_{x,i}})$ queries with certainty.

Letting $t_i = \sqrt{j_{x,i}}$ and $T = \sum_{i \in [n]} t_i^2$, we get an instance of a variable search problem.
Now fix any value of $T$ and examine only inputs with such $T$.
Denote $f$ restricted on $T$ by $f_T$.
If the quantum query complexity of $f_T$ is $\Q(f_T)$, then any algorithm that solves variable time search must require at least $\Omega(\Q(f_T))$ time.
In the following, we will prove that
there are $n=\Theta\lr{\frac{T}{\log T}}$ and 
$m=\Theta\lr{\sqrt{\frac{T}{\log T}}}$ for which
$\Q(f_T) = \Omega(\sqrt{T \log T})$.

\paragraph{Adversary bound.} We will use the relational version of the quantum adversary bound \cite{Aar06}.
Let $X \subseteq f_T^{-1}(0)$ and $Y \subseteq f_T^{-1}(1)$ and $R : X \times Y \to \mathbb R_{\geq 0}$ be a weight function.
For any input $x \in X$, define $w(x) = \sum_{y \in Y} R(x,y)$ and for any $i \in [n]$, $j \in [m]$, define $w(x,i,j) = \sum_{y \in Y, x_{i,j} \neq y_{i,j}} R(x,y)$.
Similarly define $w(y)$ and $w(y,i,j)$.
Then
\[\Q(f_T) = \Omega\Biggl(\min_{\substack{x \in x, y \in Y \\ i \in [n], j \in [m] \\ R(x,y) > 0 \\ x_{i,j} \neq y_{i,j}}} \sqrt{\frac{w(x)w(y)}{w(x,i,j)w(y,i,j)}}\Biggr).\]

\paragraph{Input sets.} Here we define the subsets of inputs $X$ and $Y$.
First, let $k$ be the smallest positive integer such that $T \leq 2^k k$ and $k$ is a multiple of $4$.
Denote $d = 2^k$, then $k = \log_2 d$ and $T = \Theta(d \log d)$.
An input $z$ from either $X$ or $Y$ must then satisfy the following conditions.
\begin{itemize}
    \item for each $p \in \left[0,\frac{k}{2}\right]$, there are exactly $\frac{d}{2^p}$ indices $i$ such that $j_{z,i} \in [2^p,2^{p+1})$; we will call the set of such indices the \emph{$p$-th block} of $z$;
    \item moreover, for each $p$ and each $\ell \in [0,2^p)$, there are exactly $\frac{d}{2^{2p}}$ indices $i$ such that $j_{z,i} = 2^p+\ell$.
\end{itemize}
Consequently, we examine inputs with $n = 2^k + 2^{k-1} + \ldots 2^{\frac{k}{2}}$ and $m = 2^{\frac{k}{2}+1}-1$.
Additionally, an input $y$ belongs to $Y$ only if there is a unique $i$ such that $y_{i,j_{y,i}} = 1$.
For this $i$, we also require $j_{y,i} \geq 2^{\frac{k}{4}+1}$: equivalently this means that $i$ belongs to a block with $p > \frac{k}{4}$.

We verify the value of $T' = \sum_{i\in[n]} t_i^2$ for these inputs.
If $i$ belongs to the $p$-th block of an input $z$, then $j_{z,i} = \Theta(2^p)$, as $j_{z,i} \in [2^p, 2^{p+1})$.
Then
\[T' = \sum_{i\in[n]} t_i^2 = \sum_{i\in[n]} j_{z,i} = \sum_{p \in \left[0,\frac{k}{2}\right]} \frac{d}{2^p}\cdot 2^p = d \left(\frac{k}{2}+1\right) = \Theta(T).\]
Note that since $\Q(f_{T'}) \leq \Q(f_T)$, a lower bound on $\Q(f_{T'})$ in terms of $T$ will also give us a lower bound on $\Q(f_T)$.
In the remainder of the proof, we will thus lower bound $\Q(f_{T'})$.

\paragraph{Relation.}

For an index $i \in [n]$ of an input $z$ that belongs to the $p$-th block, we define an \emph{index weight} $W_{z,i} = 2^p$.
Then we also define values
\begin{itemize}
    \item $B_p = \frac{d}{2^p} \cdot 2^p = d$ is the total index weight of the $p$-th block;
    \item $J_p = \frac{d}{2^{2p}} \cdot 2^p = \frac{d}{2^p}$ is the total index weight in the $p$-th block for any $j_{z,i} \in [2^p,2^{p+1})$.
\end{itemize}
Note that these values do not depend on the input.

For the relation, we will call the $p$-th block \emph{light} if $p \in \left[ 0, \frac{k}{4}\right]$ and \emph{heavy} if $p \in \left( \frac{k}{4}, \frac{k}{2} \right]$.
Two inputs $x \in X$ and $y \in Y$ have $R(x,y) > 0$ iff:
\begin{itemize}
    \item there are exactly two indices $i_0, i_1 \in [n]$ such that $j_{x,i_b} \neq j_{y,i_b}$;
    \item $i_0$ is from some light block $p_0$ and $i_1$ is from some heavy block $p_1$ of $y$; let $j_0 = j_{y,i_0}$ and $j_1 = j_{y,i_1}$;
    \item $y_{i_0,j_0} = 0$, $y_{i_1,j_1} = 1$.
    \item $x_{i_0,j_1} = x_{i_1,j_0} = 0$.
\end{itemize}
Then let the weight in the relation be
\[R(x,y) = W_{y,i_0}W_{y,i_1} = W_{x,i_1}W_{x,i_0} = 2^{p_0} 2^{p_1}.\]
Figure \ref{fig:relation} illustrates the structure of the inputs and the relation.

\begin{figure}[H]
    \begin{center}
        \includegraphics[scale=0.85]{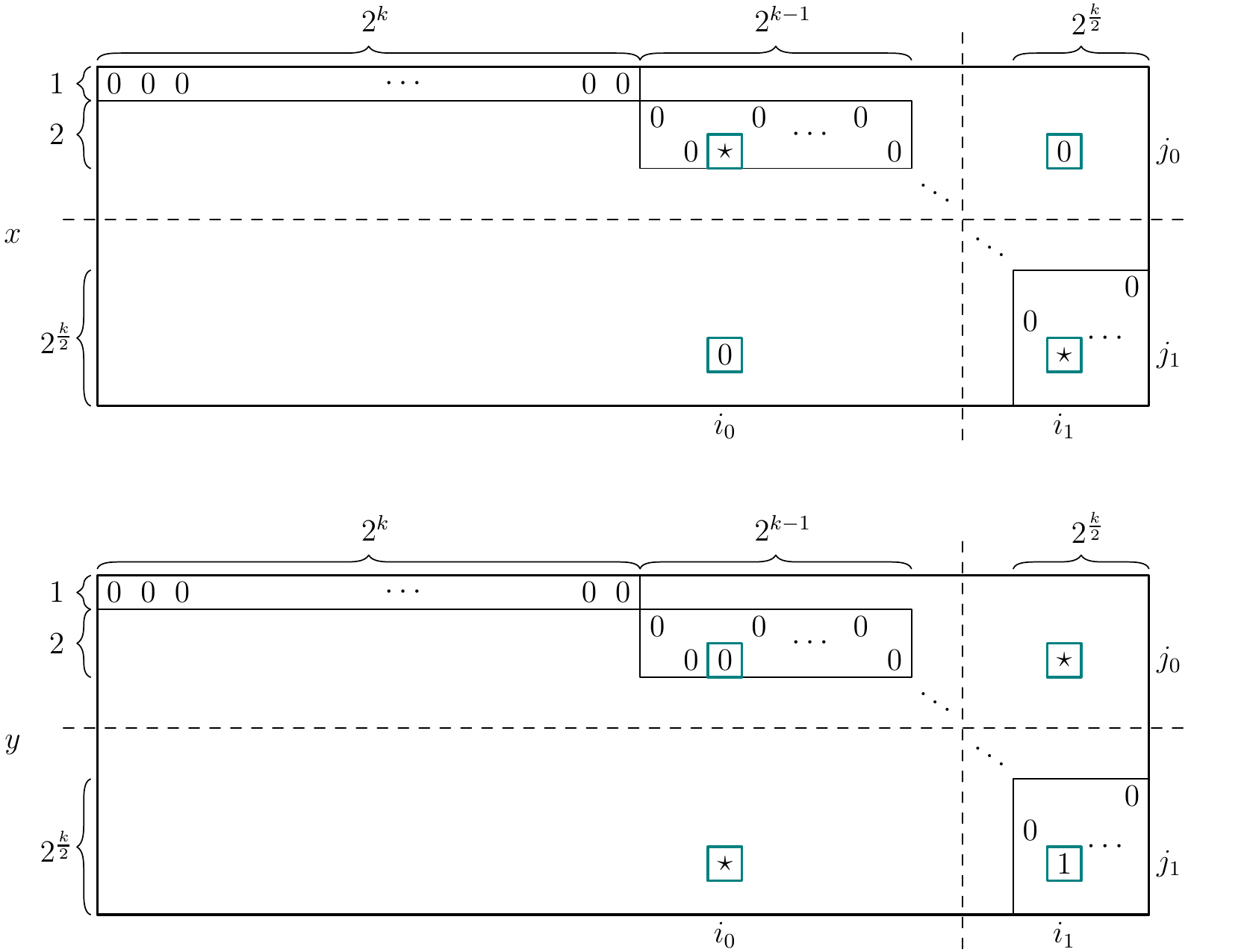}
        \caption{An example of two inputs $x \in X$ and $y \in Y$ in the relation. Inputs $x$ and $y$ differ only in the 4 highlighted positions. All of the empty cells contain $\star$, not shown for readability. For $y$, the non-$\star$ symbols of the light blocks are located in the left upper area separated by the dashed lines, while the non-$\star$ symbols of the heavy blocks are located in the lower right area. Note that for $x$, $i_0$ is in a heavy block and $i_1$ is in a light block.}
        \label{fig:relation}
    \end{center}
\end{figure}

\paragraph{Lower bound.} Now we will calculate the values for the adversary bound.
Fix two inputs $x \in X$ and $y \in Y$ with $R(x,y) > 0$.
First, since for $x$ the index $i_1$ can be any index from any light block and $i_0$ can be any index from any heavy block,
\[w(x)=\Biggl(\sum_{p_0 \in \left[0,\frac{k}{4}\right]} B_{p_0}\Biggr)\cdot\Biggl(\sum_{p_1 \in \left(\frac{k}{4},\frac{k}{2}\right]} B_{p_1}\Biggr) = \Theta(d^2k^2).\]
For $w(y)$, note that $p_1$ is uniquely determined by the position of the unique symbol $1$ in $y$.
However, the choice for $i_0$ is not additionally constrained, hence
\[w(y) = \Biggl(\sum_{p_0 \in \left[0,\frac{k}{4}\right]} B_{p_0}\Biggr)\cdot2^{p_1} = \Theta(dk2^{p_1}).\]
Therefore, the nominator in the ratio in the adversary bound is
\[w(x)w(y) = \Theta(d^3k^32^{p_1}).\]

Now note the following important property: if $x_{i,j} \neq y_{i,j}$, then one of $x_{i,j}$ and $y_{i,j}$ is $\star$, and the other is either $0$ or $1$.
There are in total exactly $4$ positions $(i,j)$ where $x$ and $y$ differ.
We will examine each case separately.
\begin{enumerate}[label=(\alph*)]
    \item $i = i_0$, $j = j_0$.
    In this case $x_{i,j} = \star$ and $y_{i,j} = 0$.
    
    For $x$, $i_1$ is not fixed but $j_0$ is known and hence also $p_0$ is known. Therefore, the total index weight from the light blocks is $J_{p_0}$.
    On the other hand, the positions of $i_0$ and, therefore, also $p_1$ are fixed.
    Thus,
    \[w(x,i,j) = J_{p_0} \cdot 2^{p_1} = \frac{d}{2^{p_0}} \cdot 2^{p_1}.\]
    For $y$, both $i_0$ and $i_1$ are fixed, hence
    \[w(y,i,j) = 2^{p_0} \cdot 2^{p_1} < d,\]
    since $p_0+p_1 \leq \frac{k}{4}+\frac{k}{2} < k$.
    Overall,
    \[w(x,i,j)w(y,i,j) < \frac{d}{2^{p_0}} \cdot 2^{p_1} \cdot d = \frac{d^2\cdot 2^{p_1}}{2^{p_0}}.\]
    
    \item $i = i_0$, $j = j_1$.
    In this case $x_{i,j} = 0$ and $y_{i,j} = \star$.
    
    For $x$, now the position $i_0$ is fixed, but $i_1$ can be chosen without additional constraints.
    The index $i_0$ uniquely defines the value of $p_1$.
    Hence,
    \[w(x,i,j) = \Biggl(\sum_{p_0 \in \left[0,\frac{k}{4}\right]} B_{p_0}\Biggr)\cdot 2^{p_1} = \Theta(dk2^{p_1}).\]
    For $y$, similarly as in the previous case, we have $i_0$ and $i_1$ fixed, thus
    \[w(y,i,j) = 2^{p_0} \cdot 2^{p_1} < d.\]
    Then
    \[w(x,i,j)w(y,i,j) = O(dk2^{p_1} \cdot d) = O(d^2k2^{p_1}).\]
    
    \item $i = i_1$, $j = j_0$.
    In this case $x_{i,j} = 0$ and $y_{i,j} = \star$.
    
    For $x$, $i_1$ is fixed, so it uniquely determines $p_0$.
    The index $i_0$ can be chosen without additional restrictions.
    Hence, 
    \[w(x,i,j) = 2^{p_0}\cdot \Biggl(\sum_{p_1 \in \left(\frac{k}{4},\frac{k}{2}\right]} B_{p_1}\Biggr) = \Theta(2^{p_0}\cdot dk).\]
    For $y$, $i_0$ is not fixed but $j_0$ is fixed, which also fixes $p_0$.
    Therefore, the total index weight from the light blocks is $J_{p_0}$.
    On the other hand, $i_1$ and $p_1$ are fixed for $y$ by the position of the symbol $1$, thus
    \[w(y,i,j) = J_{p_0} \cdot 2^{p_1} = \frac{d}{2^{p_0}} \cdot 2^{p_1}.\]
    Their product is
    \[w(x,i,j)w(y,i,j) = \Theta\left(2^{p_0}\cdot dk \cdot \frac{d}{2^{p_0}} \cdot 2^{p_1}\right) = \Theta(d^2k2^{p_1}).\]
    
    \item $i = i_1$, $j = j_1$.
    In this case $x_{i,j} = \star$ and $y_{i,j} = 1$.
    
    For $x$, $i_1$ is fixed, hence $p_0$ is also fixed; $i_0$ is not fixed, but $j_1 = j$ and $p_1$ is uniquely defined.
    Hence,
    \[w(x,i,j) = 2^{p_0} \cdot J_{p_1} = 2^{p_0} \cdot \frac{d}{2^{p_1}}.\]
    For $y$, the position of the symbol $1$ must necessarily change, hence
    \[w(y,i,j) = w(y) = \Theta(dk2^{p_1}).\]
    The product then is
    \[w(x,i,j)w(y,i,j) = \Theta\left(2^{p_0} \cdot \frac{d}{2^{p_1}} \cdot dk2^{p_1}\right) = \Theta(d^2k2^{p_0}) = O(d^2k2^{p_1}),\]
    as $p_0 \leq \frac k 4 < p_1$.
\end{enumerate}
We can see that in all cases the denominator in the ratio of the adversary bound is $O(d^2k2^{p_1})$.
Therefore,
\[\frac{w(x)w(y)}{w(x,i,j)w(y,i,j)} = \Omega\left(\frac{d^3k^32^{p_1}}{d^2k2^{p_1}}\right) = \Omega(dk^2) = \Omega(d \log^2 d)\]
and since $\log T = \Theta(\log(d \log d)) = \Theta(\log d + \log \log d) = \Theta(\log d)$, we have
\[\Q(f_T) \geq \Q(f_{T'}) = \Omega\left(\sqrt{d \log^2 d}\right) = \Omega\left(\sqrt{T \log T}\right). \qedhere\]
\end{proof}

\section{Conclusion}

In this paper, we developed a new quantum algorithm and a new quantum lower bound for variable time search. 
Our quantum algorithm has complexity $O(\sqrt{T} \log n)$, compared to $O(\sqrt{T} \log^{1.5} T)$ for the 
best previously known algorithm (quantum variable time amplitude amplification \cite{ambainis2012variable} instantiated to the case of search).
It also has the advantage of being simpler than previous quantum algorithms for variable time search.
If the recursive structure is unrolled, our algorithm consists of checking algorithms $C_{T_i}$ for various times $T_i$ interleaved with Grover diffusion steps. Thus, the structure is the essentially same as for regular search and the main difference is that $C_{T_i}$ for different $i$ are substituted at different query steps.

We note that our algorithm has a stronger assumption about $T$: we assume that an upper bound estimate $T\geq \sum_{i=1}^n t_i^2$ is provided as an input to the algorithm
and the complexity depends on this estimate $T$, rather than the actual $\sum_{i=1}^n t_i^2$. Possibly, this assumption can be removed by a doubling strategy that tries values of $T$ that keep increasing by a factor of 2 but the details remain to be worked out.

Our quantum lower bound is $\Omega(\sqrt{T \log T})$ which improves over the previously known $\Omega(\sqrt{T})$ lower bound. This shows that variable time search for the ``unknown times'' case (when the times $t_1, \ldots, t_n$ are not known in advance and cannot be used to design the quantum algorithm) is more difficult than for the ``known times'' case (which can be solved with complexity $\Theta(\sqrt{T})$).

A gap between the upper and lower bounds remains but is now just a factor of $\sqrt{\log T}$. Possibly, this is due to 
the lower bound using a set of inputs for which an approximate distribution of values $t_i$ is fixed.
In such a case, the problem may be easier than in the general case, as an approximately fixed distribution of $t_i$ can be used for algorithm design.

{\bf Acknowledgments.} We thank Krišjānis Prūsis for useful discussions on the lower bound proof. The authors are grateful to the anonymous referees for the helpful comments and suggestions. This research was supported by the ERDF project 1.1.1.5/18/A/020.

\bibliographystyle{plainlink}
\bibliography{bibliography}

\begin{appendices}
\crefalias{section}{appendix}
\section{Proofs of \texorpdfstring{\cref{thm:recursiveamplitude,thm:C1-3}}{Lemmas 2 and 3}}\label{sec:appA}
\recursiveamplitude*
\begin{proof}
	
	For each $j$ express the final state of $A_j$   in the canonical basis as
	\[
	\sum_{i=1}^n \beta_{ij} \ket {i, a_{ij}},
	\]
	where $a_{ij} \in \{ 0,1,*\} $ and $a_{ij} = 0$ iff $x_i=0$ and $t_i \leq T_j$ (i.e., iff $i \notin \mc S_j$). Initially, $\beta_{i0}= n^{-1/2}$ for all $i$. Then
	\begin{equation*}
		\sin^2 \alpha_j = \sum_{i \in \mc S_j} \lrv{\beta_{ij}}^2, 
	\end{equation*}
	for all $j$. 
	To see how the amplitude $\beta_{i(j+1)}$ is related to $\beta_{ij}$, consider how the state evolves under $A_{j+1}$:
	\begin{itemize}
		\item the final state of $A_{j}$ is 
		\[
		\sum_{i \in [n] \setminus   \mc S_{j} }
		\beta_{ij} \ket {i, 0}
		+
		\sum_{i \in \mc S_{j} }
		\beta_{ij}\ket {i, a_{ij}},
		\]
		by the definition of $\beta_{ij}$; moreover, $a_{ij}\in \{ 1,*\}$ for all $i \in \mc S_{j} $.
		
		\item Amplitude amplification $A_{j}(1)$ results in the state
		\[
		\sum_{i \in [n] \setminus   \mc S_{j} }
		\frac{\cos(3\alpha_{j})}{\cos \alpha_{j}} \beta_{ij} \ket {i, 0}
		+
		\sum_{i \in \mc S_{j} }
		\frac{\sin(3\alpha_{j})}{\sin \alpha_{j}} \beta_{ij}\ket {i, a_{ij}}
		.
		\]
		
		\item An application of $C_{T_{j+1}}$ transforms  this state to
		\[
		\sum_{i \in [n] \setminus   \mc S_{j} }
		\frac{\cos(3\alpha_{j})}{\cos \alpha_{j}} \beta_{ij} \ket {i, 0}
		+
		\sum_{i \in \mc S_{j} \setminus  \mc S_{j+1} }
		\frac{\sin(3\alpha_{j})}{\sin \alpha_{j}} \beta_{ij}\ket {i, 0}
		+
		\sum_{i \in   \mc S_{j+1} }
		\frac{\sin(3\alpha_{j})}{\sin \alpha_{j}} \beta_{ij}\ket {i, a_{i(j+1)}}.
		\]
	\end{itemize}

	We conclude that
	\begin{equation}\label{eq:nextBetaJ}
		\beta_{i(j+1)}
		=
		\begin{cases}
			\beta_{ij}\frac{\sin(3\alpha_{j})}{\sin \alpha_{j}} , 
			&  i \in  \mc S_{j}
			, \\
			\beta_{ij} \frac{\cos(3\alpha_{j})}{\cos \alpha_{j}}, 
			& i \in [n] \setminus   \mc S_{j}.
		\end{cases}
	\end{equation}
	In particular, for any $j \in [\ell]$ and $i \in \mc S_j $ we have
	\begin{equation*}
		\beta_{ij} = 
		\frac{1}{\sqrt n}\prod_{k=1}^{j-1}  \frac{\sin (3\alpha_{k})}{\sin\alpha_{k}},
	\end{equation*} 
	since   each such $i$  is in  $\mc S_k$, $k\leq j-1$, thus, by \eqref{eq:nextBetaJ}, the respective amplitude gets multiplied by $ \frac{\sin(3\alpha_{k})}{\sin \alpha_{k}} $ at each step. This establishes the second part of the lemma (that the amplitudes $\beta_{ij}$ are all equal for any $i \in \mc S_j$). For the first part, we arrive at 
	\begin{equation*}
		\sin^2 \alpha_j 
		= \sum_{i \in \mc S_j} \lrv{\beta_{ij}}^2
		=  \sum_{i \in \mc S_j} 
		\prod_{k=1}^{j-1}  \lr{\frac{\sin (3\alpha_{k})}{\sin\alpha_{k}}}^2
		=
		\frac{s_j}{n}
		\prod_{k=1}^{j-1}  \lr{\frac{\sin (3\alpha_{k})}{\sin\alpha_{k}}}^2.
	\end{equation*} 
\end{proof}

\cOneThree*
\begin{proof}

	We will  prove the following inequality:
	\begin{equation}\label{eq:sumTillSL}
		\sum_{j=1}^{\ell-1} s_j 9^j < \frac{9n}{4} .
	\end{equation}
	Then \ref{enum:1} will immediately follow, since each term on \eqref{eq:sumTillSL} is nonnegative. Furthermore, also \ref{enum:2} follows from \eqref{eq:sumTillSL} via  
	{the} generalized Bernoulli's inequality:
	\begin{align*}
		&
		\prod_{j=1}^{\ell-1}
		\lr{ 1 - \frac{4 s_j}{27n} \cdot 9^j }
		\geq
		1 - \frac{4}{27n}\sum_{j=1}^{\ell-1} s_j 9^j
		\geq 1 -\frac{4}{27n} \cdot \frac{9n}{4} = \frac{2}{3}.
	\end{align*}
	
	First we    observe that 
	\begin{equation*}
		\sum_{j=1}^d \sum_{i \in \mc S_{j-1} \setminus \mc S_j } t_i^2
		=
		\sum_{i\in [n]\setminus \mc M} t_i^2 
		<
		\sum_{i\in [n] } t_i^2  \leq T.
	\end{equation*}
	Notice that each set difference  $\mc S_{j-1} \setminus \mc S_j $
	can be characterized as follows:
	\begin{equation*}
		\mc S_{j-1} \setminus \mc S_j = 	\lrb{ i \in [n] \ : \   \lr{T_{j-1} < t_i \leq T_j} \land x_i =0}.
	\end{equation*}
	Therefore all $t_i^2$ s.t.~$i \in \mc S_{j-1} \setminus \mc S_j $ satisfy the bound
	\begin{equation*}
		t_i^2  \geq 
		T_{j-1}^2 =
		\begin{cases}
			\frac{9^{j-1} T}{n}, & j>1, \\
			0, & j=1.
		\end{cases}
	\end{equation*}
	Thus we obtain the following inequality:
	\begin{equation*}
		\frac{T}{n}\sum_{j=2}^d 
		9^{j-1}  \lrv{\mc S_{j-1} \setminus \mc S_j}
		<
		\sum_{j=1}^d \sum_{i \in \mc S_{j-1} \setminus \mc S_j } t_i^2
		< T
	\end{equation*}
or
	\begin{equation}\label{eq:skIneq}
		\sum_{k=1}^{d-1} 9^k\lr{s_k - s_{k+1}} <  n.
	\end{equation}
	We also expand  $9^\ell s_\ell$ as follows, taking into account $s_d=m$:
	\begin{equation*}
		{9^\ell s_\ell} 
		=
		9^{\ell}  \lr{s_{\ell} - s_{\ell+1} } 
		+
		\frac{1}{9} \cdot 9^{\ell+1} \lr{s_{\ell+1} - s_{\ell+2} }
		+
		\ldots
		+
		\frac{1}{9^{d-1-\ell}} \cdot 
		9^{d-1}  \lr{s_{d-1} - s_{d} }
		+
		9^\ell m.
	\end{equation*}
	From this equality, taking into account $ s_k - s_{k+1} \geq 0 $, we can upper bound $9^\ell s_\ell  $ as
	\begin{equation}\label{eq:9lsl} 
		9^\ell s_\ell 
		\leq 
		\sum_{k=\ell}^{d-1} 9^k\lr{s_k - s_{k+1}} + 9^{\ell}  m 
	\end{equation}

	Rewrite \eqref{eq:skIneq} as
	\begin{align*}
		& 
		s_1 + \frac{8}{9}\sum_{k=1}^{\ell-1} 9^k s_k-9^{\ell-1} s_\ell 
		+
		\sum_{k=\ell}^{d-1} 9^k\lr{s_k - s_{k+1}}
		<n
	\end{align*}
	and apply \eqref{eq:9lsl} to obtain 
	\begin{align*}
		& 
		s_1 + \frac{8}{9}\sum_{k=1}^{\ell-1} 9^k s_k
		+
		\sum_{k=\ell}^{d-1} 9^k\lr{s_k - s_{k+1}}
		< n + 9^{\ell-1} s_\ell 
		\leq 
		n + \frac{1}{9} \sum_{k=\ell}^{d-1} 9^k\lr{s_k - s_{k+1}} + 9^{\ell-1}  m  
		\\
		&
		\frac{8}{9}\sum_{k=1}^{\ell-1} 9^k s_k
		+
		\frac{8}{9} \sum_{k=\ell}^{d-1} 9^k\lr{s_k - s_{k+1}}
		< n -s_1 + 9^{\ell-1}  m  
		\\
		& 
		8 \sum_{k=1}^{\ell-1} 9^k s_k  
		<
		9n - 9s_1 +9^{\ell}m
		<
		9n +9^{\ell}m.
	\end{align*}
	By the choice of $\ell$ we have $9^{\ell-1} \leq \frac{n}{m} $, therefore we arrive at 
	\begin{equation*}
		8 \sum_{k=1}^{\ell-1} 9^k s_k   
		<
		9n +9 \frac{n}{m} \cdot m = 18n,
	\end{equation*}
	which is equivalent to \eqref{eq:sumTillSL}.
	
	Finally, to show \ref{enum:3}, we recall that $ s_\ell \geq s_d = m $.  Again by the choice of $ \ell $,  $ 9^\ell \geq \frac{n}{m} $. Consequently,
	\begin{equation*}
		9^\ell s_\ell 
		\geq 
		\frac{n}{m} \cdot m= n,
	\end{equation*}
	as claimed. 
\end{proof}
\end{appendices}

\end{document}